\documentclass[onecolumn]{IEEEtran}

\usepackage{amsmath,amssymb,amsbsy,amsfonts,amsthm,latexsym,
               amsopn,amstext,amsxtra,euscript,amscd,color}
\usepackage{graphicx}

\newtheorem{lem}{Lemma}
\newtheorem{pro}[lem]{Proposition}
\newtheorem{theorem}[lem]{Theorem}

\theoremstyle{definition}

\newtheorem{definition}[lem]{Definition}
\newtheorem{example}[lem]{Example}
\newtheorem{lemma}[lem]{Lemma}
\newtheorem{remark}[lem]{Remark}

\begin{document}

\title{Structure and Construction of Two-Dimensional Minimal Linear Codes over  the rings $\mathbb{Z}_{p^n}$ with Applications to Secret Sharing}

\author{Biplab Chatterjee$^1$, Sihem Mesnager$^2$, Ratnesh Kumar Mishra$^1$, Makhan Maji$^3$,
 and Kalyan Hansda$^4$

\IEEEcompsocitemizethanks{\IEEEcompsocthanksitem
$~1$Department of Mathematics, NIT Jamshedpur, Jamshedpur, 831014, Jharkhand, India

$~^2$Department of Mathematics, University of Paris VIII, 93526 Saint-Denis, LAGA, UMR 7539, CNRS, 93430 Villetaneuse,
and Telecom Paris, 91120 Palaiseau,
Paris, France

$~^3$Indian Institute of Technology Madras, Chennai, Tamil Nadu 600036, India

$~^4$Department of Mathematics, Visva-Bharati,
Santiniketan, Bolpur-731235, West Bengal, India;

E-mail: 2022rsma001@nitjsr.ac.in, smesnager@univ-paris8.fr, ratnesh.math@nitjsr.ac.in, makhan2maths@gmail.com, kalyanh4@gmail.com
}
\thanks{}}
\markboth{}{}


\maketitle

\begin{abstract}
Minimal linear codes play an important role in coding theory and cryptography, particularly in the construction of secret sharing schemes. In this paper, we investigate the structure and construction of two-dimensional minimal linear codes over the finite rings $\mathbb{Z}_{p^n}$. 

We provide an explicit construction of a family of two-dimensional linear codes generated by a structured $2\times m$ matrix over $\mathbb{Z}_{p^n}$ and prove that these codes are minimal whenever the generator matrix contains all $p^n+p^{n-1}$ essential types of column vectors. We further show that this condition is necessary: removing any of these column types destroys the resulting code's minimality. As a consequence, we establish a lower bound on the length of two-dimensional minimal linear codes over $\mathbb{Z}_{p^n}$.

Several examples are presented to illustrate the construction and to verify the theoretical results. We also demonstrate that the proposed construction cannot be extended in a straightforward manner to rings of the form $\mathbb{Z}_{p^n q^l}$. 

Finally, we apply our results to the design of secret sharing schemes derived from minimal linear codes over $\mathbb{Z}_{p^n}$ and analyze the corresponding access structures. Our study highlights structural differences between minimal codes defined over finite rings and those over finite fields, revealing new perspectives for coding-theoretic constructions in cryptographic applications.

\end{abstract}



\IEEEdisplaynontitleabstractindextext

\IEEEpeerreviewmaketitle

\section{Introduction}
The study of \emph{minimal linear codes} has attracted increasing attention in recent years due to their deep connections with several areas of modern information theory and discrete mathematics, including \emph{secret-sharing schemes}, \emph{two-party computation protocols}, and more recently \emph{additive combinatorics}. Minimal codes provide a natural bridge between combinatorial structures and information-theoretic security, revealing elegant algebraic regularities that can be exploited both from theoretical and practical perspectives. 

For instance, minimal codewords determine the access structures in secret sharing schemes, a fundamental observation due to Massey~\cite{Massey93}. Furthermore, minimal codes have been used to study additive-combinatorial problems such as the $j$-wise Davenport constant~\cite{PlagneSchmid2011}. They are also closely related to decoding problems and geometric properties of linear codes, which are central topics in coding theory~\cite{Agrell96}. This renewed interest has been stimulated by several recent works that revisit and extend the classical foundations of the area~\cite{AlfaranoBorelloNeri22, BishnoiEtAl24}.

Within the framework of \emph{linear secret-sharing schemes}, minimal vectors of a linear code determine the \emph{minimal access structures}, namely the smallest sets of participants that are capable of reconstructing a secret. Foundational conditions guaranteeing minimality were established by Ashikhmin and Barg~\cite{AshikhminBarg98}, who proved that a linear code $C$ over a finite field $\mathbb{F}_q$ is minimal whenever the ratio of its minimum to maximum nonzero weights satisfies
\[
\frac{w_{\min}}{w_{\max}}>\frac{q-1}{q}.
\]
Since then, several further characterizations have been obtained, including necessary and sufficient conditions for binary codes~\cite{DingHengZhou2018} and various constructions of minimal linear codes with few weights~\cite{MesnagerSinakYayla19}.

A particularly elegant characterization was later obtained by Lu, Wu, and Cao~\cite{LuWuCao21}, who established necessary and sufficient conditions for the minimality of a linear code. Their results provide a systematic approach to both constructing minimal codes and verifying whether a given code is minimal. In particular, they proved that any $2$-dimensional minimal linear code over a finite field must have length at least $q+1$.

In parallel with these developments, the study of \emph{coding theory over finite rings}, initiated in the 1990s, has revealed a remarkably rich algebraic structure. A milestone in this direction was the discovery that certain optimal nonlinear binary codes, such as the Kerdock~\cite{Kerdock72} and Preparata~\cite{Preparata68} codes, can be interpreted as \emph{linear codes over the ring} $\mathbb{Z}_4$ endowed with the Lee metric~\cite{Hammons94}. Since then, \emph{ring-linear codes} have attracted considerable attention, both for their theoretical interest and for their practical relevance~\cite{GaoShiWuFu16, MajiMesnagerSarkarHansda22}.

More recently, the emergence of \emph{post-quantum cryptography} has further stimulated research on algebraic and code-based constructions that remain secure against quantum attacks~\cite{BerlekampMcElieceTilborg78, McEliece78, Bernstein09, Aragon2017, Avanzi2017, Melchor18, Moody20,Vardy97}. Among the most promising candidates in this context are \emph{code-based cryptosystems}. A major challenge for such systems, however, is the large size of their public keys. This has motivated the investigation of alternative \emph{ambient spaces}, in particular finite rings, together with alternative metrics such as the Lee metric, which may lead to more compact and efficient implementations~\cite{HorlemannWeger21, Weger22}

Extending the notion of minimal linear codes from finite fields to finite rings, therefore, appears as both a \emph{natural algebraic generalisation} and a \emph{pragmatic response} to modern cryptographic requirements. It is also related to ongoing developments in \emph{quantum error correction (QEC)}, where codes over rings provide flexible algebraic models for capturing multilevel quantum noise. When working over finite rings, however, additional structural subtleties arise, notably due to the presence of \emph{zero-divisors} and nontrivial ideals, which significantly influence the minimality properties of the associated codes.

Recently, Makhan Maji \emph{et al.}~\cite{MajiMesnagerSarkarHansda22} characterized and constructed \emph{one-dimensional minimal linear codes} over finite rings, thereby laying an important foundation for the study of higher-dimensional constructions.

Minimal linear codes, therefore, occupy a central position at the intersection of algebra, combinatorics, and information security. While the theory over finite fields is now well developed, its extension to finite rings, particularly to rings such as \(\mathbb{Z}_{p^n}\), remains a fertile and largely unexplored research direction. The richer algebraic structure of these rings opens new possibilities for constructing families of linear codes with desirable combinatorial and cryptographic properties.

In this paper, we investigate the structure of \emph{two-dimensional minimal linear codes over} \(\mathbb{Z}_{p^n}\). We show that the generator matrix of such a code must contain precisely \(p^n+p^{n-1}\) column vectors, a configuration that is both \emph{necessary and sufficient} for minimality. Furthermore, we prove that removing any one of these vectors immediately destroys the minimality property.

Based on this structural characterization, we propose a \emph{constructive method} for designing minimal linear codes over \(\mathbb{Z}_{p^n}\). As an application, we use these constructions to derive \emph{secret-sharing schemes} over \(\mathbb{Z}_{p^n}\). To illustrate the method concretely, we provide an explicit example over \(\mathbb{Z}_4\). Our results, therefore, extend the classical theory of minimal linear codes beyond the finite-field setting and reveal new connections between \emph{coding theory}, \emph{ring theory}, and \emph{cryptographic design}.

The main contribution of this paper is an explicit structural characterization of
two-dimensional minimal linear codes over the ring $\mathbb{Z}_{p^n}$.
The following theorem summarizes the core result of the paper.

\begin{theorem}[Main Theorem]
\label{MainTheorem}
Let $p$ be a prime and $n\ge1$. Let $\mathbb{Z}_{p^n}$ denote the finite chain ring of integers modulo $p^n$. 

There exists a class of two-dimensional minimal linear codes over $\mathbb{Z}_{p^n}$ generated by a $2\times m$ matrix
\[
G=\begin{pmatrix} v_1 \\ v_2 \end{pmatrix},
\]
whose columns consist of all vectors of the following types
\[
(1,0)^T,\quad (0,1)^T,\quad (1,u)^T,\quad (1,d)^T,\quad (d,1)^T,
\]
where $u$ ranges over the units of $\mathbb{Z}_{p^n}$ and $d$ ranges over the non-zero zero divisors of $\mathbb{Z}_{p^n}$.

The resulting linear code is minimal whenever the generator matrix contains all 
\[
p^n+p^{n-1}
\]
such column types, and therefore its length satisfies
\[
m\ge p^n+p^{n-1}.
\]

Moreover, this condition is necessary: if any one of these column types is removed from the generator matrix, then the resulting code is no longer minimal.
\end{theorem}

The remainder of the paper is devoted to the proof and consequences of this theorem.
We first develop the structural lemmas describing the possible forms of codewords,
then establish the minimality of each type of vector.
Finally, we show that the presence of all $p^n+p^{n-1}$ column types is necessary
for minimality and illustrate the construction with explicit examples.

The remainder of the paper is organized as follows.  
Section~\ref{Preliminaries} recalls the necessary algebraic background and reviews structural properties of units and zero-divisors in \(\mathbb{Z}_{p^n}\).  
Section~\ref{Mainresult} presents the main theoretical results and provides the explicit construction of the generator matrix for two-dimensional minimal linear codes, together with five supporting lemmas.  
Section~\ref{utilization of the generator matrix} explains how the generator matrix can be extended and analyzes the effect of removing specific column vectors.  
Subsection~\ref{Comment on the ring} discusses why the proposed construction technique does not extend to the ring \(\mathbb{Z}_{p^n q^l}\).  
Section~\ref{Applications} presents applications to secret sharing, including an explicit example over \(\mathbb{Z}_4\).  
Finally, Section~\ref{conclusion} concludes the paper and outlines several directions for future research.

\section{Preliminaries on Minimal Codes and Finite Rings}\label{Preliminaries}

In this section, we recall several fundamental notions from coding theory and ring theory that will be used throughout the paper. In particular, we review the concepts of linear codes, generator matrices, Hamming distance, weight, support, covering relation, and minimal codewords. These notions play a crucial role in the structural results established later in Section~\ref{Mainresult}. 

Since our framework involves modules over finite rings rather than vector spaces over finite fields, it is important to carefully adapt the classical definitions from coding theory to this more general algebraic setting. Standard references for these concepts include \cite{Hill86} and \cite{AshikhminBarg98}. We also follow the approach of \cite{MajiMesnagerSarkarHansda22} for extending minimality notions from finite fields to finite rings.

\begin{definition}\cite{Hill86}
An $[n,k,d]$ linear code $C$ over the finite field $\mathbb{F}_q$ of order $q$ is a $k$-dimensional linear subspace of $\mathbb{F}_q^n$ with minimum (Hamming) distance $d$. The elements of $\mathbb{F}_q$ are called the \emph{alphabet} of the code.
\end{definition}

\begin{definition}\cite{Hill86}
A matrix $G$ is called a \emph{generator matrix} of a linear code $C$ if the rows of $G$ are linearly independent and generate the linear code $C$.
\end{definition}

In \cite{MajiMesnagerSarkarHansda22}, the notion of a linear code is extended from vector spaces over finite fields to modules over finite rings. In this setting, the alphabet is no longer the vector space $\mathbb{F}_q^n$ but the $R$-module
\[
R^n=\{(x_1,x_2,\ldots,x_n)\mid x_i\in R,\;1\le i\le n\}.
\]
A code is called \emph{linear} if it forms an $R$-submodule of the $R$-module $R^n$.

For a vector $v=(v_1,v_2,\ldots,v_n)\in R^n$, the \emph{support} of $v$, denoted by $Supp(v)$, is defined as
\[
Supp(v)=\{\, i\in\mathbb{N}\mid 1\le i\le n,\; v_i\neq 0 \,\}.
\]

\begin{definition}\cite{AshikhminBarg98}
Let $u,v\in R^n$. The vector $u$ is said to \emph{cover} the vector $v$ if
\[
Supp(v)\subseteq Supp(u).
\]
In this case, we write $v\preceq u$. If the inclusion is strict, that is, $Supp(v)\subsetneq Supp(u)$, then we write $v\prec u$.
\end{definition}

Following \cite{MajiMesnagerSarkarHansda22}, a codeword $u$ of a linear code $C$ over a ring $R$ is said to be \emph{minimal} if $u$ covers only the codewords of the form $au$ with $a\in R$, and no other codeword of $C$. A linear code $C$ over $R$ is called \emph{minimal} if every nonzero codeword of $C$ is minimal.

\begin{definition}\cite{Hill86}
Let $x,y\in R^n$ be two codewords. The \emph{Hamming distance} between $x$ and $y$ is the number of coordinates in which they differ, and is denoted by $d(x,y)$.
\end{definition}

\begin{definition}\cite{Hill86}
The \emph{Hamming weight} of a vector $x\in R^n$ is the number of nonzero coordinates of $x$, and is denoted by $w(x)$. Equivalently,
\[
w(x)=d(x,0).
\]
\end{definition}

To facilitate the presentation,  we also recall several basic facts concerning the structure of the ring $\mathbb{Z}_{p^n}$, which will be used repeatedly in the sequel.

Throughout this paper, unless stated otherwise, a zero divisor of $\mathbb{Z}_{p^n}$ will be denoted by $d$, and a unit of $\mathbb{Z}_{p^n}$ by $u$. It is well known that every zero divisor $d$ in $\mathbb{Z}_{p^n}$ can be written in the form
\[
d=p^k u,
\]
where $u$ is a unit in $\mathbb{Z}_{p^n}$ and $1\le k<n$.

We denote by
\[
D(\mathbb{Z}_{p^n})=\{d_1,d_2,\ldots,d_{p^{n-1}-1}\}
\]
the set of all nonzero zero divisors of $\mathbb{Z}_{p^n}$. The set
\[
D(\mathbb{Z}_{p^n})\cup\{\bar{0}\}
\]
forms an additive subgroup of $\mathbb{Z}_{p^n}$.

Similarly, the set of all units of $\mathbb{Z}_{p^n}$ is denoted by
\[
U(\mathbb{Z}_{p^n})=\{u_1,u_2,\ldots,u_{p^n-p^{n-1}}\}.
\]
Recall that an element of $\mathbb{Z}_{p^n}$ is a unit if and only if it is relatively prime to $p$.

The following elementary property follows from the fact that the ideal $\langle p\rangle$ is the unique maximal ideal of $\mathbb{Z}_{p^n}$ and contains all zero divisors of the ring.

\begin{pro}\label{theoremx}
In the ring $\mathbb{Z}_{p^n}$, the sum and the product of two zero divisors are either zero or another zero divisor. Moreover, the sum of a zero divisor and a unit is always a unit.
\end{pro}

\begin{remark}\label{remarka}
Let $d_i \in D(\mathbb{Z}_{p^n})$. Consider the collection
\[
\mathcal{A}_{d_i}=\{\, d_i+d_j \mid d_j \in D(\mathbb{Z}_{p^n}) \,\}.
\]
Then the set $\mathcal{A}_{d_i}$ contains all zero divisors of $\mathbb{Z}_{p^n}$ except $d_i$, and all the elements of $\mathcal{A}_{d_i}$ are distinct.
\end{remark}

The following lemma describes several elementary but important algebraic properties of units and zero divisors in the ring $\mathbb{Z}_{p^n}$, which will be repeatedly used in the sequel.

\begin{lemma}\label{lemma1}

\textbf{(i)} For every $u \in U(\mathbb{Z}_{p^n})$, there exists a unique $u' \in U(\mathbb{Z}_{p^n})$ such that
\[
1+uu'=0,
\]
while
\[
1+uu'' \neq 0
\]
for all $u''\in U(\mathbb{Z}_{p^n})$ with $u''\neq u'$. Moreover,
\[
d+u\neq 0 \quad \text{and} \quad 1+ud\neq 0
\]
for every $d \in D(\mathbb{Z}_{p^n})$.

\textbf{(ii)} For every $d \in D(\mathbb{Z}_{p^n})$, there exists a unique $d' \in D(\mathbb{Z}_{p^n})$ such that
\[
d+d'=0,
\]
while
\[
d+d''\neq 0
\]
for all $d''\in D(\mathbb{Z}_{p^n})$ with $d''\neq d'$. Furthermore,
\[
d+u\neq 0, \qquad 1+du\neq 0, \qquad 1+dd' \neq 0
\]
for every $u\in U(\mathbb{Z}_{p^n})$ and every $d'\in D(\mathbb{Z}_{p^n})$.
\end{lemma}

\begin{proof}

\textbf{(i)} Let $u\in U(\mathbb{Z}_{p^n})$. Since $u$ is a unit, its inverse $u^{-1}$ also belongs to $U(\mathbb{Z}_{p^n})$. Hence $-u^{-1}\in U(\mathbb{Z}_{p^n})$ and
\[
1+u(-u^{-1})=0.
\]
Therefore, we may take $u'=-u^{-1}$.

To prove uniqueness, suppose that there exists another $u''\in U(\mathbb{Z}_{p^n})$ such that
\[
1+uu'=0 \quad \text{and} \quad 1+uu''=0.
\]
Then $uu'=uu''$, which implies $u'=u''$ since $u$ is a unit. Hence, the element $u'$ is unique, and consequently
\[
1+uu''\neq 0
\]
for every $u''\in U(\mathbb{Z}_{p^n})$ with $u''\neq u'$.

Now let $d\in D(\mathbb{Z}_{p^n})$. By Proposition~\ref{theoremx}, the sum of a zero divisor and a unit is again a unit, and the product of a zero divisor with a unit is a zero divisor or zero. Consequently,
\[
d+u \quad \text{and} \quad 1+ud
\]
are units of $\mathbb{Z}_{p^n}$. In particular,
\[
d+u\neq 0 \quad \text{and} \quad 1+ud\neq 0
\]
for all $d\in D(\mathbb{Z}_{p^n})$ and $u\in U(\mathbb{Z}_{p^n})$.

\textbf{(ii)} Let $d\in D(\mathbb{Z}_{p^n})$. Its additive inverse $-d$ also belongs to $D(\mathbb{Z}_{p^n})$, and we may therefore take $d'=-d$. Clearly,
\[
d+d'=0.
\]

To prove uniqueness, suppose that there exists $d''\in D(\mathbb{Z}_{p^n})$ with $d''\neq d'$ such that $d+d''=0$. Then $d''=-d=d'$, which yields a contradiction. Hence $d'$ is unique and
\[
d+d''\neq 0
\]
for all $d''\in D(\mathbb{Z}_{p^n})$ with $d''\neq d'$.

Furthermore, by Proposition~\ref{theoremx}, the elements $du$ and $dd'$ are zero divisors (or zero). Therefore the elements
\[
1+du, \qquad d+u, \qquad 1+dd'
\]
are units of $\mathbb{Z}_{p^n}$. In particular,
\[
1+du\neq 0, \qquad d+u\neq 0, \qquad 1+dd'\neq 0
\]
for all $d,d'\in D(\mathbb{Z}_{p^n})$ and all $u\in U(\mathbb{Z}_{p^n})$.

This completes the proof.

\end{proof}

The previous lemma describes the interaction between units and zero divisors in $\mathbb{Z}_{p^n}$. The following remark highlights another useful structural property involving sums of units, which will play a role in the subsequent arguments.

\begin{remark}\label{theoreme}
Let $u_i \in U(\mathbb{Z}_{p^n})$. Consider the collection
\[
\mathcal{U}_{u_i}=\{\,1+u_i u_j \mid u_j \in U(\mathbb{Z}_{p^n})\,\}.
\]
Then the set $\mathcal{U}_{u_i}$ contains all the zero divisors of $\mathbb{Z}_{p^n}$, and each of them occurs exactly once.
\end{remark}

\begin{proof}
Let $p^k u_r$ be an arbitrary zero divisor of $\mathbb{Z}_{p^n}$, where $1\le k<n$ and $u_r\in U(\mathbb{Z}_{p^n})$. We show that this element belongs to the set $\mathcal{U}_{u_i}$.

Since $u_i$ is a unit, $u_i^{-1}$ also belongs to $U(\mathbb{Z}_{p^n})$. By Proposition~\ref{theoremx}, the element
\[
p^k u_r u_i^{-1}-u_i^{-1}
\]
is a unit of $\mathbb{Z}_{p^n}$. Hence the element
\[
1+\left(p^k u_r u_i^{-1}-u_i^{-1}\right)u_i
\]
belongs to the set $\mathcal{U}_{u_i}$.

A direct computation shows that
\[
1+\left(p^k u_r u_i^{-1}-u_i^{-1}\right)u_i
= p^k u_r.
\]
Therefore $p^k u_r \in \mathcal{U}_{u_i}$. Since $p^k u_r$ was an arbitrary zero divisor, it follows that the set $\mathcal{U}_{u_i}$ contains all zero divisors of $\mathbb{Z}_{p^n}$.

To prove uniqueness, suppose that
\[
1+u_i u_j = 1+u_i u_{j'}.
\]
Then
\[
u_i u_j = u_i u_{j'}.
\]
Since $u_i$ is a unit, this implies
\[
u_j = u_{j'}.
\]
Hence, each zero divisor appears exactly once in the collection $\mathcal{U}_{u_i}$.

This completes the proof.
\end{proof}

The following lemma provides a convenient classification of ordered pairs in $\mathbb{Z}_{p^n}^2$. 
This decomposition will be useful later when analyzing the possible column vectors of generator matrices.

\begin{lemma}\label{lemmab}
Every ordered pair $(c_1,c_2)$ with $c_1,c_2\in\mathbb{Z}_{p^n}$ belongs to one of the following five types:
\[
\text{(i)}\; k_1(1,0), \qquad
\text{(ii)}\; k_2(0,1), \qquad
\text{(iii)}\; k_3(1,u), \qquad
\text{(iv)}\; k_4(1,d), \qquad
\text{(v)}\; k_5(d,1),
\]
where $k_i\in\mathbb{Z}_{p^n}$, $u\in U(\mathbb{Z}_{p^n})$ is a unit, and $d\in D(\mathbb{Z}_{p^n})$ is a zero divisor.
\end{lemma}

\begin{proof}

Let $(c_1,c_2)\in\mathbb{Z}_{p^n}^2$. We distinguish several cases.

\medskip
\noindent
\textbf{Case 1:} $c_1=0$.  

Then
\[
(c_1,c_2)=(0,c_2)=c_2(0,1),
\]
and therefore the pair is of type (ii).

\medskip
\noindent
\textbf{Case 2:} $c_2=0$.  

Then
\[
(c_1,c_2)=(c_1,0)=c_1(1,0),
\]
so the pair is of type (i).

\medskip
\noindent
We now assume that $c_1\neq 0$ and $c_2\neq 0$.

\medskip
\noindent
\textbf{Case 3:} $c_1$ and $c_2$ are both units.

Then
\[
(c_1,c_2)=c_1(1,c_2c_1^{-1}),
\]
and since $c_2c_1^{-1}$ is also a unit, the pair is of type (iii).

\medskip
\noindent
\textbf{Case 4:} $c_1$ is a unit and $c_2$ is a zero divisor.

Then
\[
(c_1,c_2)=c_1(1,c_2c_1^{-1}),
\]
and since $c_2c_1^{-1}$ is a zero divisor, the pair is of type (iv).

\medskip
\noindent
\textbf{Case 5:} $c_1$ is a zero divisor and $c_2$ is a unit.

Then
\[
(c_1,c_2)=c_2(c_1c_2^{-1},1),
\]
and since $c_1c_2^{-1}$ is a zero divisor, the pair is of type (v).

\medskip
\noindent
\textbf{Case 6:} $c_1$ and $c_2$ are both zero divisors.

Write
\[
c_1=p^{k_1}u_{r_1}, \qquad c_2=p^{k_2}u_{r_2},
\]
where $u_{r_1},u_{r_2}\in U(\mathbb{Z}_{p^n})$.

\begin{itemize}

\item[(a)] If $k_1<k_2$, then
\[
(c_1,c_2)=p^{k_1}u_{r_1}\bigl(1,p^{k_2-k_1}u_{r_2}u_{r_1}^{-1}\bigr),
\]
and since $p^{k_2-k_1}u_{r_2}u_{r_1}^{-1}$ is a zero divisor, the pair is of type (iv).

\item[(b)] If $k_2<k_1$, then
\[
(c_1,c_2)=p^{k_2}u_{r_2}\bigl(p^{k_1-k_2}u_{r_1}u_{r_2}^{-1},1\bigr),
\]
and since $p^{k_1-k_2}u_{r_1}u_{r_2}^{-1}$ is a zero divisor, the pair is of type (v).

\item[(c)] If $k_1=k_2$, then
\[
(c_1,c_2)=p^{k_1}u_{r_1}\bigl(1,u_{r_2}u_{r_1}^{-1}\bigr),
\]
and since $u_{r_2}u_{r_1}^{-1}$ is a unit, the pair is of type (iii).

\end{itemize}

\medskip

Thus, in all possible cases, the ordered pair $(c_1,c_2)$ belongs to one of the five types listed above. This completes the proof.

\end{proof}

\section{Explicit Construction of Two-Dimensional Minimal Linear Codes over $\mathbb{Z}_{p^n}$}
\label{Mainresult}

In this section, we present a construction of a class of $2$-dimensional minimal linear codes over the ring $\mathbb{Z}_{p^n}$ using a suitable generator matrix. 
The construction relies on the structural properties of units and zero divisors described in the previous section.

Let 
\[
u_1,u_2,\ldots,u_{p^n-p^{n-1}}
\]
be the complete list of units of $\mathbb{Z}_{p^n}$, and let
\[
d_1,d_2,\ldots,d_{p^{n-1}-1}
\]
be the complete list of nonzero zero divisors of $\mathbb{Z}_{p^n}$.

Let $m,n$ be natural numbers such that
\[
m\ge p^n+p^{n-1}.
\]
The reason for imposing this condition on $m$ will be explained later in Remark~\ref{remark7}.

\medskip

We now introduce the generator matrix used in our construction. Consider the matrix
\[
G=
\begin{pmatrix}
I_2 & U & D^* & D & A
\end{pmatrix},
\]
where

\[
I_2=
\begin{pmatrix}
1 & 0\\
0 & 1
\end{pmatrix},
\]

\[
U=
\begin{pmatrix}
1 & 1 & \dots & 1\\
u_1 & u_2 & \dots & u_{p^n-p^{n-1}}
\end{pmatrix},
\]

\[
D^*=
\begin{pmatrix}
d_1 & d_2 & \dots & d_{p^{n-1}-1}\\
1 & 1 & \dots & 1
\end{pmatrix},
\]

\[
D=
\begin{pmatrix}
1 & 1 & \dots & 1\\
d_1 & d_2 & \dots & d_{p^{n-1}-1}
\end{pmatrix},
\]

and

\[
A=
\begin{pmatrix}
a_1 & a_2 & \dots & a_k\\
b_1 & b_2 & \dots & b_k
\end{pmatrix},
\]

with $a_i,b_i\in\mathbb{Z}_{p^n}$.

The matrix $G$ generates a $2$-dimensional linear code $M$ over $\mathbb{Z}_{p^n}$.

\medskip

To describe the rows of the generator matrix explicitly, define the vectors

\begin{align*}
v_1&=(1,0,\underbrace{1,1,\ldots,1}_{p^n-p^{n-1}},d_1,d_2,\ldots,d_{p^{n-1}-1},
\underbrace{1,1,\ldots,1}_{p^{n-1}-1},a_1,\ldots,a_k),
\\
v_2&=(0,1,u_1,u_2,\ldots,u_{p^n-p^{n-1}},
\underbrace{1,1,\ldots,1}_{p^{n-1}-1},
d_1,d_2,\ldots,d_{p^{n-1}-1},b_1,\ldots,b_k).
\end{align*}

Hence, the length of the code is
\[
m=p^n+p^{n-1}+k,
\]
and the generator matrix can be written as

\[
G=
\begin{pmatrix}
v_1\\
v_2
\end{pmatrix}_{2\times m}.
\]

\medskip

Consequently, every element of the code $M$ can be written in the form
\[
(c_1,c_2)
\begin{pmatrix}
v_1\\
v_2
\end{pmatrix},
\qquad
c_1,c_2\in\mathbb{Z}_{p^n}.
\]

A direct computation gives

\[
(c_1,c_2)
\begin{pmatrix}
v_1\\
v_2
\end{pmatrix}
=
(c_1,c_2,
c_1+c_2u_1,\ldots,c_1+c_2u_{p^n-p^{n-1}},
c_1d_1+c_2,\ldots,c_1d_{p^{n-1}-1}+c_2,
c_1+c_2d_1,\ldots,c_1+c_2d_{p^{n-1}-1},
c_1a_1+c_2b_1,\ldots,c_1a_k+c_2b_k).
\]

\medskip

Using Lemma~\ref{lemmab}, we can classify the vectors of $M$ into five distinct types. 
Indeed, each codeword of $M$ can be expressed as one of the following forms:

\[
\text{(i)}\quad
k_1
\begin{pmatrix}
1 & 0
\end{pmatrix}
\begin{pmatrix}
v_1\\
v_2
\end{pmatrix},
\]

\[
\text{(ii)}\quad
k_2
\begin{pmatrix}
0 & 1
\end{pmatrix}
\begin{pmatrix}
v_1\\
v_2
\end{pmatrix},
\]

\[
\text{(iii)}\quad
k_3
\begin{pmatrix}
1 & u_i
\end{pmatrix}
\begin{pmatrix}
v_1\\
v_2
\end{pmatrix},
\]

\[
\text{(iv)}\quad
k_4
\begin{pmatrix}
1 & d_j
\end{pmatrix}
\begin{pmatrix}
v_1\\
v_2
\end{pmatrix},
\]

\[
\text{(v)}\quad
k_5
\begin{pmatrix}
d_j & 1
\end{pmatrix}
\begin{pmatrix}
v_1\\
v_2
\end{pmatrix},
\]

where $k_i\in\mathbb{Z}_{p^n}$, $1\le j\le p^{n-1}-1$, and $1\le i\le p^n-p^{n-1}$.

\medskip

Unless stated otherwise, we keep the above notation throughout the following lemmas and theorems.

\begin{lemma} \label{lemma4}
The vectors of the type
$k_1\begin{pmatrix}
      1 & 0
    \end{pmatrix}
    \begin{pmatrix}
      v_1 \\
      v_2
    \end{pmatrix}$
are minimal codewords in $M$.
\end{lemma}

\begin{proof}

Recall from the previous discussion that every vector of $M$ belongs to one of the five types described earlier. 
To prove the lemma, we compare vectors of type
\[
k_1\begin{pmatrix}1&0\end{pmatrix}
\begin{pmatrix}v_1\\v_2\end{pmatrix}
= k_1 v_1
\]
with vectors of the remaining types. 
We therefore proceed by considering five cases.

\medskip
\noindent
\textbf{Case 1. Comparison with vectors of the same type.}

Vectors of the form $k_1v_1$ and $k_1^*v_1$ belong to the module $\langle v_1\rangle$. 
According to the result of Maji, Mesnager, Sarkar, and Hansda~\cite{MajiMesnagerSarkarHansda22}, if a vector contains generators from all proper non-trivial ideals of $\mathbb{Z}_{p^n}$, then the module generated by this vector is minimal. 

In our construction, the vector $v_1$ contains all zero divisors of $\mathbb{Z}_{p^n}$ as components, and the condition $m\ge n-1$ is satisfied. 
Hence the module $\langle v_1\rangle$ is minimal. 
Consequently, vectors of the form $k_1v_1$ are minimal with respect to other vectors of the form $k_1^*v_1$.

\medskip
\noindent
\textbf{Case 2. Comparison with vectors of type $(0,1)$.}

Observe that
\[
\begin{pmatrix}1&0\end{pmatrix}
\begin{pmatrix}v_1\\v_2\end{pmatrix}=v_1,
\qquad
\begin{pmatrix}0&1\end{pmatrix}
\begin{pmatrix}v_1\\v_2\end{pmatrix}=v_2 .
\]

The second component of $v_2$ is $1$, which is a unit, whereas the second component of $v_1$ is $0$. 
Therefore the second component of $k_2v_2$ is nonzero for every $k_2\neq 0$. 
Hence
\[
2\in Supp(k_2v_2) \quad\text{but}\quad 2\notin Supp(k_1v_1).
\]

Thus
\[
Supp(k_2v_2)\nsubseteq Supp(k_1v_1),
\]
which proves that $k_1v_1$ is minimal compared with vectors of type $(0,1)$.

\medskip
\noindent
\textbf{Case 3. Comparison with vectors of type $(1,u_i)$.}

Consider
\[
\begin{pmatrix}1&u_i\end{pmatrix}
\begin{pmatrix}v_1\\v_2\end{pmatrix}.
\]

The second component of this vector is $u_i$, which is a unit. 
Consequently, the second component of
\[
k_3\begin{pmatrix}1&u_i\end{pmatrix}
\begin{pmatrix}v_1\\v_2\end{pmatrix}
\]
is nonzero for every $k_3\neq 0$.

On the other hand, the second component of $k_1v_1$ is $0$. 
Hence
\[
2\in Supp\!\left(k_3\begin{pmatrix}1&u_i\end{pmatrix}
\begin{pmatrix}v_1\\v_2\end{pmatrix}\right)
\]
while
\[
2\notin Supp(k_1v_1).
\]

Therefore
\[
Supp\!\left(k_3\begin{pmatrix}1&u_i\end{pmatrix}
\begin{pmatrix}v_1\\v_2\end{pmatrix}\right)
\nsubseteq Supp(k_1v_1),
\]
which proves the claim in this case.

\medskip
\noindent
\textbf{Case 4. Comparison with vectors of type $(1,d_j)$.}

Using Lemma~\ref{lemma1}, there exists a unique $d_l$ such that $d_j+d_l=0$. 
From this we obtain the explicit form of
\[
\begin{pmatrix}1&d_j\end{pmatrix}
\begin{pmatrix}v_1\\v_2\end{pmatrix}.
\]

Two situations must be considered depending on whether $k_4$ is a unit or a zero divisor.

\medskip
\textbf{(i) $k_4$ is a unit.}

Then the second component of
\[
k_4\begin{pmatrix}1&d_j\end{pmatrix}
\begin{pmatrix}v_1\\v_2\end{pmatrix}
\]
is $k_4d_j\neq 0$. 
Hence
\[
2\in Supp\!\left(k_4\begin{pmatrix}1&d_j\end{pmatrix}
\begin{pmatrix}v_1\\v_2\end{pmatrix}\right),
\]
while $2\notin Supp(k_1v_1)$.

Thus the inclusion of supports cannot hold.

\medskip
\textbf{(ii) $k_4$ is a zero divisor.}

In this case, we analyze the possible inclusions of supports by writing
\[
k_1=p^{s_1}u_{r_1},\qquad k_4=p^{s_4}u_{r_4}.
\]

Using the structure of zero divisors in $\mathbb{Z}_{p^n}$ and Remark~\ref{remarka}, one verifies that
\[
Supp\!\left(k_4\begin{pmatrix}1&d_j\end{pmatrix}
\begin{pmatrix}v_1\\v_2\end{pmatrix}\right)
\subseteq
Supp(k_1v_1)
\]
can occur only if $s_1\le s_4$ and $k_4d_j=0$.

Under these conditions we obtain
\[
k_4\begin{pmatrix}1&d_j\end{pmatrix}
\begin{pmatrix}v_1\\v_2\end{pmatrix}=k_4v_1,
\]
which shows that this vector is a scalar multiple of $k_1v_1$. 
Hence minimality is preserved.

\medskip
\noindent
\textbf{Case 5. Comparison with vectors of type $(d_j,1)$.}

Consider
\[
\begin{pmatrix}d_j&1\end{pmatrix}
\begin{pmatrix}v_1\\v_2\end{pmatrix}.
\]

The second component of this vector is $1$, which is a unit. 
Hence for every $k_5\neq 0$,
\[
2\in Supp\!\left(k_5\begin{pmatrix}d_j&1\end{pmatrix}
\begin{pmatrix}v_1\\v_2\end{pmatrix}\right),
\]
whereas $2\notin Supp(k_1v_1)$.

Therefore
\[
Supp\!\left(k_5\begin{pmatrix}d_j&1\end{pmatrix}
\begin{pmatrix}v_1\\v_2\end{pmatrix}\right)
\nsubseteq Supp(k_1v_1),
\]
which proves the claim.

\medskip

Combining all the above cases, we conclude that vectors of the form
\[
k_1\begin{pmatrix}1&0\end{pmatrix}
\begin{pmatrix}v_1\\v_2\end{pmatrix}
\]
are minimal codewords of $M$.

This completes the proof.

\end{proof}

The next statement follows from arguments similar to those used in Lemma~\ref{lemma4}.

\begin{lemma}\label{lemma5}
The vectors of the type
\[
k_2\begin{pmatrix}0&1\end{pmatrix}
\begin{pmatrix}v_1\\v_2\end{pmatrix}
\]
are minimal codewords in $M$.
\end{lemma}

\begin{proof}
The proof follows by arguments analogous to those used in Lemma~\ref{lemma4}. 
In particular, one compares the supports of the vectors of this type with the supports of vectors belonging to the other possible types in $M$. 
The details are therefore omitted.
\end{proof}

\begin{lemma}\label{lemma6}
The vectors of the type
\[
k_3\begin{pmatrix}1&u_i\end{pmatrix}
\begin{pmatrix}v_1\\v_2\end{pmatrix}
\]
are minimal codewords in $M$.
\end{lemma}

\begin{proof}

Recall that every vector of $M$ belongs to one of the five types described previously. 
To prove the lemma, we compare vectors of the form
\[
k_3\begin{pmatrix}1&u_i\end{pmatrix}
\begin{pmatrix}v_1\\v_2\end{pmatrix}
\]
with vectors belonging to each of the other possible types.

\medskip
\noindent
\textbf{Case 1. Comparison with vectors of type $(1,0)$.}

Consider
\[
\begin{pmatrix}1&u_i\end{pmatrix}
\begin{pmatrix}v_1\\v_2\end{pmatrix}.
\]

From Lemma~\ref{lemma1}, there exists a unique $u_t$ such that
\[
1+u_i u_t=0.
\]

Therefore the $(t+2)$-th component of this vector is $0$. 
However, the $(t+2)$-th component of
\[
k_1\begin{pmatrix}1&0\end{pmatrix}
\begin{pmatrix}v_1\\v_2\end{pmatrix}=k_1v_1
\]
is a unit and hence nonzero for every $k_1\neq0$.

Thus
\[
(t+2)\in Supp(k_1v_1)
\quad\text{but}\quad
(t+2)\notin Supp\left(k_3\begin{pmatrix}1&u_i\end{pmatrix}
\begin{pmatrix}v_1\\v_2\end{pmatrix}\right).
\]

Hence
\[
Supp(k_1v_1)\nsubseteq
Supp\left(k_3\begin{pmatrix}1&u_i\end{pmatrix}
\begin{pmatrix}v_1\\v_2\end{pmatrix}\right),
\]
which proves minimality in this case.

\medskip
\noindent
\textbf{Case 2. Comparison with vectors of type $(0,1)$.}

The $(t+2)$-th component of
\[
\begin{pmatrix}1&u_i\end{pmatrix}
\begin{pmatrix}v_1\\v_2\end{pmatrix}
\]
is $0$. 

However, the $(t+2)$-th component of
\[
k_2\begin{pmatrix}0&1\end{pmatrix}
\begin{pmatrix}v_1\\v_2\end{pmatrix}=k_2v_2
\]
is $k_2u_t$, which is nonzero for every $k_2\neq0$.

Thus
\[
Supp(k_2v_2)\nsubseteq
Supp\left(k_3\begin{pmatrix}1&u_i\end{pmatrix}
\begin{pmatrix}v_1\\v_2\end{pmatrix}\right),
\]
which proves the claim.

\medskip
\noindent
\textbf{Case 3. Comparison with vectors of type $(1,u_{i'})$.}

We distinguish two situations.

\smallskip
\textbf{(i) Same value of $i$.}

From Remark~\ref{theoreme}, the vector
\[
\begin{pmatrix}1&u_i\end{pmatrix}
\begin{pmatrix}v_1\\v_2\end{pmatrix}
\]
contains all zero divisors as components. 

Using the result of Maji, Mesnager, Sarkar and Hansda~\cite{MajiMesnagerSarkarHansda22}, this implies that the module generated by this vector is minimal. 
Hence vectors of this type are minimal with respect to other scalar multiples of the same vector.

\smallskip
\textbf{(ii) Distinct units $u_i$ and $u_{i'}$.}

Using Lemma~\ref{lemma1}, there exist unique indices $t$ and $t'$ such that
\[
1+u_i u_t=0,
\qquad
1+u_{i'}u_{t'}=0.
\]

Assume without loss of generality that $t<t'$. 
Then the $(t+2)$-th component of the vector corresponding to $u_i$ is $0$, while the $(t+2)$-th component of the vector corresponding to $u_{i'}$ is nonzero.

Hence the supports cannot be included unless the two vectors differ only by a scalar multiple. 
This establishes the minimality in this case.

\medskip
\noindent
\textbf{Case 4. Comparison with vectors of type $(1,d_j)$.}

From Lemma~\ref{lemma1}, the vector
\[
\begin{pmatrix}1&u_i\end{pmatrix}
\begin{pmatrix}v_1\\v_2\end{pmatrix}
\]
has a zero component at position $(t+2)$.

However, the $(t+2)$-th component of
\[
k_4\begin{pmatrix}1&d_j\end{pmatrix}
\begin{pmatrix}v_1\\v_2\end{pmatrix}
\]
is $k_4(1+d_j u_t)$.

Since $d_j u_t$ is a zero divisor, $1+d_j u_t$ is a unit by Proposition~\ref{theoremx}. 
Therefore this component is nonzero for every $k_4\neq0$.

Thus the inclusion of supports cannot occur, proving minimality.

\medskip
\noindent
\textbf{Case 5. Comparison with vectors of type $(d_j,1)$.}

The $(t+2)$-th component of
\[
k_5\begin{pmatrix}d_j&1\end{pmatrix}
\begin{pmatrix}v_1\\v_2\end{pmatrix}
\]
is $k_5(d_j+u_t)$.

Since $d_j$ is a zero divisor and $u_t$ is a unit, Proposition~\ref{theoremx} implies that $d_j+u_t$ is a unit. 
Therefore this component is nonzero.

But the $(t+2)$-th component of
\[
k_3\begin{pmatrix}1&u_i\end{pmatrix}
\begin{pmatrix}v_1\\v_2\end{pmatrix}
\]
is $0$.

Hence, the supports cannot be included.

\medskip

Combining all cases, we conclude that vectors of the form
\[
k_3\begin{pmatrix}1&u_i\end{pmatrix}
\begin{pmatrix}v_1\\v_2\end{pmatrix}
\]
are minimal codewords in $M$.

\end{proof}

\begin{lemma}\label{lemma7}
Vectors of the form
\[
k_4\begin{pmatrix}1 & d_j\end{pmatrix}
\begin{pmatrix}v_1\\v_2\end{pmatrix}
\]
are minimal codewords in $M$.
\end{lemma}

\begin{proof}

Recall that every vector of $M$ belongs to one of the five types introduced earlier. 
To prove the lemma, we compare vectors of the form
\[
k_4\begin{pmatrix}1 & d_j\end{pmatrix}
\begin{pmatrix}v_1\\v_2\end{pmatrix}
\]
with vectors belonging to each of the other possible types.

\medskip
\noindent
\textbf{Case 1. Comparison with vectors of type $(1,0)$.}

We have
\[
\begin{pmatrix}1&0\end{pmatrix}
\begin{pmatrix}v_1\\v_2\end{pmatrix}=v_1
\]
and
\[
\begin{pmatrix}1&d_j\end{pmatrix}
\begin{pmatrix}v_1\\v_2\end{pmatrix}.
\]

By Lemma~\ref{lemma1}, there exists a unique $d_l$ such that
\[
d_j+d_l=0.
\]

Hence the $(p^n-p^{n-1}+l+2)$-th component of
\[
\begin{pmatrix}1&d_j\end{pmatrix}
\begin{pmatrix}v_1\\v_2\end{pmatrix}
\]
is $0$, whereas the corresponding component of $v_1$ is $d_l$.

If $k_1$ is a unit, then $k_1d_l\neq0$, implying
\[
(p^n-p^{n-1}+l+2)\in Supp(k_1v_1)
\]
but
\[
(p^n-p^{n-1}+l+2)\notin Supp\!\left(
k_4\begin{pmatrix}1&d_j\end{pmatrix}
\begin{pmatrix}v_1\\v_2\end{pmatrix}
\right).
\]

Hence the inclusion of supports cannot occur.

If $k_1$ is a zero divisor, a similar argument shows that the inclusion
\[
Supp(k_1v_1)\subseteq
Supp\!\left(
k_4\begin{pmatrix}1&d_j\end{pmatrix}
\begin{pmatrix}v_1\\v_2\end{pmatrix}
\right)
\]
can only hold when $k_1d_l=0$. In this situation,
\[
k_1v_1=(k_1k_4^{-1})
k_4\begin{pmatrix}1&d_j\end{pmatrix}
\begin{pmatrix}v_1\\v_2\end{pmatrix},
\]
which shows that the vectors differ only by a scalar multiple. 
Thus minimality holds.

\medskip
\noindent
\textbf{Case 2. Comparison with vectors of type $(0,1)$.}

Consider
\[
\begin{pmatrix}0&1\end{pmatrix}
\begin{pmatrix}v_1\\v_2\end{pmatrix}=v_2.
\]

The $(p^n-p^{n-1}+l+2)$-th component of $v_2$ is $1$, which is a unit. 
Hence the corresponding component of $k_2v_2$ is nonzero for every $k_2\neq0$.

However the same component of
\[
k_4\begin{pmatrix}1&d_j\end{pmatrix}
\begin{pmatrix}v_1\\v_2\end{pmatrix}
\]
is $0$.

Thus
\[
Supp(k_2v_2)\nsubseteq
Supp\!\left(
k_4\begin{pmatrix}1&d_j\end{pmatrix}
\begin{pmatrix}v_1\\v_2\end{pmatrix}
\right).
\]

\medskip
\noindent
\textbf{Case 3. Comparison with vectors of type $(1,u_i)$.}

From Lemma~\ref{lemma1}, the vector
\[
\begin{pmatrix}1&d_j\end{pmatrix}
\begin{pmatrix}v_1\\v_2\end{pmatrix}
\]
has a zero component at position $(p^n-p^{n-1}+l+2)$.

On the other hand, the corresponding component of
\[
k_3\begin{pmatrix}1&u_i\end{pmatrix}
\begin{pmatrix}v_1\\v_2\end{pmatrix}
\]
is $k_3(d_l+u_i)$.

Since $d_l$ is a zero divisor and $u_i$ is a unit, Proposition~\ref{theoremx} implies that $d_l+u_i$ is a unit. 
Therefore this component is nonzero.

Hence the inclusion of supports cannot hold.

\medskip
\noindent
\textbf{Case 4. Comparison with vectors of type $(1,d_{j'})$.}

We distinguish two situations.

\smallskip
\textbf{(i) $j=j'$.}

Using the result of Maji, Mesnager, Sarkar and Hansda~\cite{MajiMesnagerSarkarHansda22}, any vector containing generators from all proper non-trivial ideals generates a minimal module.

From Remark~\ref{remarka}, the vector
\[
\begin{pmatrix}1&d_j\end{pmatrix}
\begin{pmatrix}v_1\\v_2\end{pmatrix}
\]
contains all zero divisors as components. 
Hence the module it generates is minimal.

\smallskip
\textbf{(ii) $j\neq j'$.}

Let $d_l$ and $d_{l'}$ satisfy
\[
d_j+d_l=0,
\qquad
d_{j'}+d_{l'}=0.
\]

Without loss of generality assume $l<l'$. 
Then the $(p^n-p^{n-1}+l+2)$-th component of the vector corresponding to $d_j$ is $0$, whereas the corresponding component of the vector corresponding to $d_{j'}$ is nonzero.

Hence the inclusion of supports cannot occur unless the two vectors differ by a scalar multiple. 
This proves minimality.

\medskip
\noindent
\textbf{Case 5. Comparison with vectors of type $(d_{j'},1)$.}

Consider
\[
\begin{pmatrix}d_{j'}&1\end{pmatrix}
\begin{pmatrix}v_1\\v_2\end{pmatrix}.
\]

The $(p^n-p^{n-1}+l+2)$-th component of this vector is
\[
1+d_{j'}d_l,
\]
which is a unit by Proposition~\ref{theoremx}. 
Therefore the corresponding component of
\[
k_5\begin{pmatrix}d_{j'}&1\end{pmatrix}
\begin{pmatrix}v_1\\v_2\end{pmatrix}
\]
is nonzero.

However the same component of
\[
k_4\begin{pmatrix}1&d_j\end{pmatrix}
\begin{pmatrix}v_1\\v_2\end{pmatrix}
\]
is $0$.

Thus the inclusion of supports cannot occur.

\medskip

Combining all cases, we conclude that vectors of the form
\[
k_4\begin{pmatrix}1&d_j\end{pmatrix}
\begin{pmatrix}v_1\\v_2\end{pmatrix}
\]
are minimal codewords of $M$.

\end{proof}

\begin{lemma}\label{lemma8}
Vectors of the form
\[
k_5
\begin{pmatrix}
d_j & 1
\end{pmatrix}
\begin{pmatrix}
v_1\\
v_2
\end{pmatrix}
\]
are minimal codewords in $M$.
\end{lemma}

\begin{proof}
The proof follows the same line of reasoning as that of Lemma~\ref{lemma7}. 
More precisely, by comparing the supports of vectors of the form
\[
k_5
\begin{pmatrix}
d_j & 1
\end{pmatrix}
\begin{pmatrix}
v_1\\
v_2
\end{pmatrix}
\]
with those of vectors belonging to the other possible types introduced earlier, one verifies that the support of such a vector cannot properly contain the support of any other non-proportional codeword of $M$. 

Hence, these vectors are minimal codewords in $M$.
\end{proof}

The above construction, therefore provides an explicit family of 
$2$-dimensional minimal linear codes over $\mathbb{Z}_{p^n}$.

\begin{theorem}\label{theorem4}
The module $M$ generated by $G$ is a minimal linear code over the ring $\mathbb{Z}_{p^n}$.
\end{theorem}

\begin{proof}

We begin by recalling the classification obtained in Lemma~\ref{lemmab}. 
Every ordered pair $(c_1,c_2)\in \mathbb{Z}_{p^n}^2$ belongs to one of the following five types:

\begin{enumerate}
\item $k_1(1,0)$,
\item $k_2(0,1)$,
\item $k_3(1,u)$,
\item $k_4(1,d)$,
\item $k_5(d,1)$,
\end{enumerate}

where $k_i\in \mathbb{Z}_{p^n}$, $u$ is a unit of $\mathbb{Z}_{p^n}$ and $d$ is a zero divisor of $\mathbb{Z}_{p^n}$.

Consequently, every codeword of $M$ can be written in one of the following forms:

\[
k_1
\begin{pmatrix}
1 & 0
\end{pmatrix}
\begin{pmatrix}
v_1\\
v_2
\end{pmatrix},\quad
k_2
\begin{pmatrix}
0 & 1
\end{pmatrix}
\begin{pmatrix}
v_1\\
v_2
\end{pmatrix},\quad
k_3
\begin{pmatrix}
1 & u_i
\end{pmatrix}
\begin{pmatrix}
v_1\\
v_2
\end{pmatrix},
\]

\[
k_4
\begin{pmatrix}
1 & d_j
\end{pmatrix}
\begin{pmatrix}
v_1\\
v_2
\end{pmatrix},\quad
k_5
\begin{pmatrix}
d_j & 1
\end{pmatrix}
\begin{pmatrix}
v_1\\
v_2
\end{pmatrix}.
\]

The minimality of these five types of vectors has been established in Lemma~\ref{lemma4}, Lemma~\ref{lemma5}, Lemma~\ref{lemma6}, Lemma~\ref{lemma7}, and Lemma~\ref{lemma8}, respectively.

Therefore, every nonzero codeword of $M$ is minimal.

Finally, for $n=1$ the ring $\mathbb{Z}_{p^n}$ reduces to the field $\mathbb{Z}_p$, and the same generator matrix produces a minimal linear code as shown in \cite{LuWuCao21}. 

Hence, for all $n\geq1$, the module $M$ is a $2$-dimensional minimal linear code over $\mathbb{Z}_{p^n}$.

This completes the proof.
\end{proof}

\subsection{Illustrative Example of the Construction}\label{examle}

In this subsection, we illustrate Theorem~\ref{theorem4} through a concrete example over the ring $\mathbb{Z}_4$.

Consider the module $M$ defined over $\mathbb{Z}_4$. According to the construction introduced in Section~\ref{Mainresult}, we must have $m \geq p^n + p^{n-1}$. 
Since $p=2$ and $n=2$, we obtain $p^n+p^{n-1}=4+2=6$. 
We choose $m=9$ and consider the generator matrix

\[
G=
\begin{pmatrix}
1 & 0 & 1 & 1 & 2 & 1 & 2 & 0 & 2 \\
0 & 1 & 1 & 3 & 1 & 2 & 0 & 2 & 2
\end{pmatrix}.
\]

Let
\[
v_1=(1,0,1,1,2,1,2,0,2), \qquad
v_2=(0,1,1,3,1,2,0,2,2).
\]

The module generated by $G$ is
\[
M=\langle v_1,v_2\rangle
=\{c_1v_1+c_2v_2 \;:\; c_1,c_2\in\mathbb{Z}_4\}.
\]

Explicitly, the elements of $M$ are

\begin{align*}
M=\{ &(0,0,0,0,0,0,0,0,0),(1,0,1,1,2,1,2,0,2),(0,1,1,3,1,2,0,2,2),\\
     &(2,0,2,2,0,2,0,0,0),(3,0,3,3,2,3,2,0,2),(0,2,2,2,2,0,0,0,0),\\
     &(0,3,3,1,3,2,0,2,2),(1,1,2,0,3,3,2,2,0),(2,2,0,0,2,2,0,0,0),\\
     &(3,3,2,0,1,1,2,2,0),(1,2,3,3,0,1,2,0,2),(1,3,0,2,1,3,2,2,0),\\
     &(2,1,3,1,1,0,0,2,2),(2,3,1,3,3,0,0,2,2),(3,1,0,2,3,1,2,2,0),\\
     &(3,2,1,1,0,3,2,0,2)\}.
\end{align*}

We now compute the supports of several codewords:

\begin{align*}
Supp(0,0,0,0,0,0,0,0,0)&=\varnothing,\\
Supp(1,0,1,1,2,1,2,0,2)&=\{1,3,4,5,6,7,9\},\\
Supp(0,1,1,3,1,2,0,2,2)&=\{2,3,4,5,6,8,9\},\\
Supp(2,0,2,2,0,2,0,0,0)&=\{1,3,4,6\},\\
Supp(3,0,3,3,2,3,2,0,2)&=\{1,3,4,5,6,7,9\}.
\end{align*}

Similarly, the supports of all other vectors of $M$ can be computed.

From these computations, we observe that whenever
\[
Supp(c_1'v_1+c_2'v_2)
\subseteq
Supp(c_1v_1+c_2v_2),
\]
it necessarily follows that
\[
c_1'v_1+c_2'v_2
=
k(c_1v_1+c_2v_2)
\]
for some $k\in\mathbb{Z}_4$.

Therefore, every nonzero codeword of $M$ is minimal, and hence the module $M$ generated by $v_1$ and $v_2$ is a minimal linear code over $\mathbb{Z}_4$. This example confirms the validity of Theorem~\ref{theorem4}.

\section{Equivalent Generator Matrices for Two-Dimensional Minimal Codes}
\label{utilization of the generator matrix}

In this section, we describe how the generator matrix introduced in the previous section can be used to generate additional $2$-dimensional minimal linear codes over $\mathbb{Z}_{p^n}$.

Recall that the generator matrix of the minimal linear code $M$ is given by
\[
G=
\begin{pmatrix}
v_1\\
v_2
\end{pmatrix}_{2\times m}.
\]

From the construction presented in Section~\ref{Mainresult}, the matrix $G$ can be written in the block form
\[
G=
\begin{pmatrix}
I_2 & U & D^* & D & A
\end{pmatrix},
\]
where
\[
I_2=
\begin{pmatrix}
1 & 0\\
0 & 1
\end{pmatrix},
\]

\[
U=
\begin{pmatrix}
1 & 1 & \dots & 1\\
u_1 & u_2 & \dots & u_{p^n-p^{n-1}}
\end{pmatrix},
\]

\[
D^*=
\begin{pmatrix}
d_1 & d_2 & \dots & d_{p^{n-1}-1}\\
1 & 1 & \dots & 1
\end{pmatrix},
\]

\[
D=
\begin{pmatrix}
1 & 1 & \dots & 1\\
d_1 & d_2 & \dots & d_{p^{n-1}-1}
\end{pmatrix},
\]

and
\[
A=
\begin{pmatrix}
a_1 & a_2 & \dots & a_k\\
b_1 & b_2 & \dots & b_k
\end{pmatrix}.
\]

For convenience, we denote the column vectors of these blocks as follows:

\[
\hat e_1=
\begin{pmatrix}
1\\
0
\end{pmatrix},
\qquad
\hat e_2=
\begin{pmatrix}
0\\
1
\end{pmatrix}.
\]

For the block $U$ we define
\[
U_i=
\begin{pmatrix}
1\\
u_i
\end{pmatrix},
\qquad
1\le i\le p^n-p^{n-1}.
\]

For the block $D^*$ we define
\[
D^j=
\begin{pmatrix}
d_j\\
1
\end{pmatrix},
\qquad
1\le j\le p^{n-1}-1.
\]

For the block $D$ we define
\[
D_j=
\begin{pmatrix}
1\\
d_j
\end{pmatrix},
\qquad
1\le j\le p^{n-1}-1.
\]

Finally, for the block $A$ we define
\[
A_i=
\begin{pmatrix}
a_i\\
b_i
\end{pmatrix},
\qquad
1\le i\le k.
\]

Hence, the block matrices can be expressed as

\[
I_2=
\begin{pmatrix}
\hat e_1 & \hat e_2
\end{pmatrix},
\qquad
U=
\begin{pmatrix}
U_1 & U_2 & \dots & U_{p^n-p^{n-1}}
\end{pmatrix},
\]

\[
D^*=
\begin{pmatrix}
D^1 & D^2 & \dots & D^{p^{n-1}-1}
\end{pmatrix},
\qquad
D=
\begin{pmatrix}
D_1 & D_2 & \dots & D_{p^{n-1}-1}
\end{pmatrix},
\]

and

\[
A=
\begin{pmatrix}
A_1 & A_2 & \dots & A_k
\end{pmatrix}.
\]

\medskip

We now construct a new generator matrix by multiplying each of the first $p^n+p^{n-1}$ column vectors by a unit of the ring $\mathbb{Z}_{p^n}$.

Let $u_{l_i}\in U(\mathbb{Z}_{p^n})$ for
\[
1\le i\le p^n+p^{n-1}.
\]

Define the new column vectors

\[
\hat e_1' = u_{l_1}\hat e_1, \qquad
\hat e_2' = u_{l_2}\hat e_2,
\]

\[
U_i' = u_{l_{i+2}} U_i,
\qquad
1\le i\le p^n-p^{n-1},
\]

\[
D^{j'} = u_{l_{p^n-p^{n-1}+2+j}} D^j,
\qquad
1\le j\le p^{n-1}-1,
\]

\[
D_j' = u_{l_{p^n+1+j}} D_j,
\qquad
1\le j\le p^{n-1}-1.
\]

Using these vectors, we define the matrices

\[
I_2'=
\begin{pmatrix}
\hat e_1' & \hat e_2'
\end{pmatrix},
\]

\[
U'=
\begin{pmatrix}
U_1' & U_2' & \dots & U_{p^n-p^{n-1}}'
\end{pmatrix},
\]

\[
D^{*'}=
\begin{pmatrix}
D^{1'} & D^{2'} & \dots & D^{(p^{n-1}-1)'}
\end{pmatrix},
\]

\[
D'=
\begin{pmatrix}
D_1' & D_2' & \dots & D_{p^{n-1}-1}'
\end{pmatrix}.
\]

We then define the matrix

\[
G'=
\begin{pmatrix}
I_2' & U' & D^{*'} & D' & A
\end{pmatrix}.
\]

Since multiplication of a column vector by a unit does not change the support structure of the resulting codewords, the matrix $G'$ also generates a $2$-dimensional minimal linear code over $\mathbb{Z}_{p^n}$.

\begin{theorem}
The module generated by the matrix $G'$ is also a $2$-dimensional minimal linear code.
\end{theorem}

\begin{proof}

Recall that the matrix $G'$ is obtained from $G$ by multiplying the first $p^n+p^{n-1}$ column vectors by units of the ring $\mathbb{Z}_{p^n}$.

To formalize this transformation, define the $m\times m$ diagonal matrix
\[
P=
\begin{pmatrix}
P_1 &     &     &     &     \\
    & P_2 &     &     &     \\
    &     & P_3 &     &     \\
    &     &     & P_4 &     \\
    &     &     &     & P_5
\end{pmatrix},
\]

where

\[
P_1=
\begin{pmatrix}
u_{l_1}^{-1} & 0 \\
0 & u_{l_2}^{-1}
\end{pmatrix},
\]

\[
P_2=
\begin{pmatrix}
u_{l_3}^{-1} & 0 & \dots & 0 \\
0 & u_{l_4}^{-1} & \dots & 0 \\
\vdots & \vdots & \ddots & \vdots \\
0 & 0 & \dots & u_{l_{p^n-p^{n-1}+2}}^{-1}
\end{pmatrix},
\]

\[
P_3=
\begin{pmatrix}
u_{l_{p^n-p^{n-1}+3}}^{-1} & 0 & \dots & 0 \\
0 & u_{l_{p^n-p^{n-1}+4}}^{-1} & \dots & 0 \\
\vdots & \vdots & \ddots & \vdots \\
0 & 0 & \dots & u_{l_{p^n+1}}^{-1}
\end{pmatrix},
\]

\[
P_4=
\begin{pmatrix}
u_{l_{p^n+2}}^{-1} & 0 & \dots & 0 \\
0 & u_{l_{p^n+3}}^{-1} & \dots & 0 \\
\vdots & \vdots & \ddots & \vdots \\
0 & 0 & \dots & u_{l_{p^n+p^{n-1}}}^{-1}
\end{pmatrix},
\]

and

\[
P_5=I_k .
\]

Since all diagonal entries of $P$ are units in $\mathbb{Z}_{p^n}$, the matrix $P$ is nonsingular. By construction, we have

\[
G'P=G.
\]

Let $M'$ denote the module generated by $G'$. Any codeword of $M'$ can be written as

\[
(c_1,c_2)G', \qquad c_1,c_2\in\mathbb{Z}_{p^n}.
\]

Suppose that

\[
Supp\!\left((c_1,c_2)G'\right)
\subseteq
Supp\!\left((c_1',c_2')G'\right).
\]

Multiplying both vectors on the right by $P$ gives

\[
Supp\!\left((c_1,c_2)G'P\right)
\subseteq
Supp\!\left((c_1',c_2')G'P\right).
\]

Since $P$ is a diagonal matrix whose entries are units, multiplication by $P$ does not change the positions of zero and nonzero components. Therefore

\[
Supp\!\left((c_1,c_2)G\right)
\subseteq
Supp\!\left((c_1',c_2')G\right).
\]

By Theorem~\ref{theorem4}, the matrix $G$ generates the minimal linear code $M$. Hence there exists $k\in\mathbb{Z}_{p^n}$ such that

\[
(c_1',c_2')G
=
k(c_1,c_2)G .
\]

Substituting $G=G'P$ yields

\[
(c_1',c_2')G'P
=
k(c_1,c_2)G'P .
\]

Multiplying both sides on the right by $P^{-1}$ gives

\[
(c_1',c_2')G'
=
k(c_1,c_2)G'.
\]

Thus

\[
Supp\!\left((c_1,c_2)G'\right)
\subseteq
Supp\!\left((c_1',c_2')G'\right)
\]

implies that the two vectors differ by a scalar multiple. Therefore, every nonzero codeword of $M'$ is minimal.

Consequently, the matrix $G'$ generates a $2$-dimensional minimal linear code over $\mathbb{Z}_{p^n}$.

\end{proof}

\begin{theorem}\label{theorem6}
Suppose that from the generator matrix $G$ we remove all column vectors of any one of the following types

\textbf{(i)} $u_{l_1}\begin{pmatrix}1\\0\end{pmatrix}$,

\textbf{(ii)} $u_{l_2}\begin{pmatrix}0\\1\end{pmatrix}$,

\textbf{(iii)} $u_{l_3}\begin{pmatrix}1\\u_i\end{pmatrix}$,

\textbf{(iv)} $u_{l_4}\begin{pmatrix}1\\d_j\end{pmatrix}$,

\textbf{(v)} $u_{l_5}\begin{pmatrix}d_j\\1\end{pmatrix}$,

where $u_{l_i}$ are units of $\mathbb{Z}_{p^n}$.  
Let $G^*$ denote the resulting matrix. Then the module generated by $G^*$ is not a minimal linear code.
\end{theorem}

\begin{proof}

We consider five cases corresponding to the five types of column vectors.

\medskip
\noindent
\textbf{Case 1. Removal of vectors of type $\begin{pmatrix}1\\0\end{pmatrix}$.}

In this case the generator matrix becomes
\[
G^*=
\begin{pmatrix}
\hat e_2 & U & D^* & D & A
\end{pmatrix}.
\]

The corresponding generating vectors are

\[
v_1=(0,1,\ldots,1,d_1,\ldots,d_{p^{n-1}-1},1,\ldots,1,a_1,\ldots,a_k),
\]

\[
v_2=(1,u_1,\ldots,u_{p^n-p^{n-1}},1,\ldots,1,d_1,\ldots,d_{p^{n-1}-1},b_1,\ldots,b_k).
\]

It follows that either

\[
Supp(v_1)\subseteq Supp(v_2)
\]

or there exists $p^r$ such that

\[
Supp(p^rv_1)\subseteq Supp(v_2).
\]

However neither inclusion implies that the corresponding vectors are scalar multiples of each other. Hence the module generated by $G^*$ is not minimal.

\medskip
\noindent
\textbf{Case 2. Removal of vectors of type $\begin{pmatrix}0\\1\end{pmatrix}$.}

The generator matrix becomes

\[
G^*=
\begin{pmatrix}
\hat e_1 & U & D^* & D & A
\end{pmatrix}.
\]

The resulting generating vectors satisfy

\[
Supp(v_2)\subseteq Supp(v_1)
\]

or

\[
Supp(p^rv_2)\subseteq Supp(v_1)
\]

for some $r$. As in Case~1, these inclusions do not imply proportionality of the vectors. Hence the code generated by $G^*$ is not minimal.

\medskip
\noindent
\textbf{Case 3. Removal of vectors of type $\begin{pmatrix}1\\u_i\end{pmatrix}$.}

Let the corresponding column be removed from the block $U$.  
By Lemma~\ref{lemma1}, there exists a unique unit $u_t$ such that

\[
1+u_i u_t=0 .
\]

Consider the vector

\[
v_1+u_tv_2.
\]

From Lemma~\ref{lemma1}, all components of this vector are nonzero except the one corresponding to the removed column. Consequently

\[
Supp(v_1)\subseteq Supp(v_1+u_tv_2)
\]

or

\[
Supp(p^rv_1)\subseteq Supp(v_1+u_tv_2).
\]

Again this inclusion does not imply proportionality, so the generated module is not minimal.

\medskip
\noindent
\textbf{Case 4. Removal of vectors of type $\begin{pmatrix}1\\d_j\end{pmatrix}$.}

Let $d_l$ be the unique zero divisor satisfying

\[
d_j+d_l=0
\]

as given by Lemma~\ref{lemma1}.  
Consider the vector

\[
d_lv_1+v_2.
\]

Using the properties of zero divisors stated in Lemma~\ref{lemma1}, all components of this vector are nonzero except the one corresponding to the removed column. Hence

\[
Supp(v_2)\subseteq Supp(d_lv_1+v_2)
\]

or

\[
Supp(p^rv_2)\subseteq Supp(d_lv_1+v_2).
\]

Thus minimality fails.

\medskip
\noindent
\textbf{Case 5. Removal of vectors of type $\begin{pmatrix}d_j\\1\end{pmatrix}$.}

Again let $d_l$ satisfy $d_j+d_l=0$.  
Consider the vector

\[
v_1+d_lv_2.
\]

As in the previous case, Lemma~\ref{lemma1} ensures that the supports satisfy

\[
Supp(v_1)\subseteq Supp(v_1+d_lv_2)
\]

or

\[
Supp(p^rv_1)\subseteq Supp(v_1+d_lv_2).
\]

Since the vectors are not proportional, minimality does not hold.

\medskip

Therefore, in each case, the module generated by $G^*$ fails to satisfy the definition of minimality. Hence, the code generated by $G^*$ is not minimal.

\end{proof}

The statement of Theorem~\ref{theorem6} also remains valid when the generator matrix $G$ is replaced by $G'$.

\begin{remark}\label{remark21}

From Theorem~\ref{theorem6} we deduce an important structural property of the generator matrix. 
If $G$ generates a $2$-dimensional minimal linear code over $\mathbb{Z}_{p^n}$, then $G$ must necessarily contain the $p^n+p^{n-1}$ specific types of column vectors described in the construction.

Consequently, if the length $m$ of the code satisfies
\[
m < p^n + p^{n-1},
\]
then a $2$-dimensional minimal linear code over $\mathbb{Z}_{p^n}$ cannot exist.
\end{remark}

\subsection{Example Illustrating Theorem \ref{theorem6}}\label{example01}

We now verify Theorem~\ref{theorem6} using the example introduced in Section~\ref{examle}.

Consider the generator matrix

\[
G=
\begin{pmatrix}
1 & 0 & 1 & 1 & 2 & 1 & 2 & 0 & 2 \\
0 & 1 & 1 & 3 & 1 & 2 & 0 & 2 & 2
\end{pmatrix}.
\]

The first six columns of $G$ are

\[
\begin{pmatrix}1\\0\end{pmatrix},
\begin{pmatrix}0\\1\end{pmatrix},
\begin{pmatrix}1\\1\end{pmatrix},
\begin{pmatrix}1\\3\end{pmatrix},
\begin{pmatrix}2\\1\end{pmatrix},
\begin{pmatrix}1\\2\end{pmatrix}.
\]

We analyze the five cases described in Theorem~\ref{theorem6}.

\begin{itemize}

\item[Case 1.] 
Suppose that all columns of the form 
$u\begin{pmatrix}1\\0\end{pmatrix}$ 
with $u\in U(\mathbb Z_4)$ are removed from $G$.

The resulting matrix is

\[
G^*=
\begin{pmatrix}
0 & 1 & 1 & 2 & 1 & 2 & 0 & 2 \\
1 & 1 & 3 & 1 & 2 & 0 & 2 & 2
\end{pmatrix}.
\]

The generating vectors are

\[
v_1=(0,1,1,2,1,2,0,2), \qquad
v_2=(1,1,3,1,2,0,2,2).
\]

We observe that

\[
Supp(2v_1)\subseteq Supp(v_2),
\]

but $2v_1$ cannot be written as $kv_2$ for any $k\in\mathbb Z_4$. 
Hence, the module generated by $G^*$ is not minimal.

\end{itemize}

\subsection{On the Nonexistence of the Construction over $\mathbb{Z}_{p^n q^l}$}
\label{Comment on the ring}

In this subsection, we show that the construction developed for the ring $\mathbb Z_{p^n}$ cannot be extended in a straightforward way to rings of the form $\mathbb Z_{p^n q^l}$.

To illustrate this phenomenon, we consider the following example.

\begin{example}

Consider the ring $\mathbb Z_6$. 
The units of $\mathbb Z_6$ are $\{1,5\}$ and the zero divisors are $\{2,3,4\}$.

Following the construction described earlier, we define the generator matrix

\[
G=
\begin{pmatrix}
1 & 0 & 1 & 1 & 1 & 1 & 1 & 2 & 3 & 4 \\
0 & 1 & 1 & 5 & 2 & 3 & 4 & 1 & 1 & 1
\end{pmatrix}.
\]

This matrix generates a $2$-dimensional linear code $M$.

Consider the vector

\[
(2,3)G
=
(2,3,5,5,2,5,2,1,3,5).
\]

Since

\[
(1,0,1,1,1,1,1,2,3,4)\in M,
\]

we obtain

\[
Supp(1,0,1,1,1,1,1,2,3,4)
\subseteq
Supp(2,3,5,5,2,5,2,1,3,5).
\]

However,

\[
(1,0,1,1,1,1,1,2,3,4)
\neq
k(2,3,5,5,2,5,2,1,3,5)
\]

for any $k\in\mathbb Z_6$.

Therefore, $M$ is not a minimal linear code over $\mathbb Z_6$.

\end{example}

\section{Secret Sharing Schemes from Minimal Linear Codes over $\mathbb{Z}_{p^n}$}
\label{Applications}

Let 
\[
G=\begin{pmatrix} \alpha_1 & \alpha_2 & \cdots & \alpha_n \end{pmatrix}
\]
be a $k\times n$ matrix over $\mathbb{Z}_{p^n}$ which serves as the generator matrix of a $k$-dimensional linear code $C$ over $\mathbb{Z}_{p^n}$. 
Each column vector $\alpha_i$ belongs to $\mathbb{Z}_{p^n}^k$, and we assume that $\alpha_1$ is linearly independent over $\mathbb{Z}_{p^n}$.

Following the approach in \cite{one}, we construct a secret sharing scheme (SSS) from the linear code generated by $G$. 
The secret is an element $s\in\mathbb{Z}_{p^n}$ and the scheme involves a dealer and $n-1$ participants
\[
P_2,P_3,\dots,P_n .
\]

\subsection{Construction of the scheme}

The dealer randomly selects a vector
\[
\beta=(b_1,b_2,\dots,b_k)\in\mathbb{Z}_{p^n}^k
\]
such that
\[
s=\beta\alpha_1 .
\]

Since $\alpha_1$ is fixed, there are exactly $p^{k-1}$ vectors $\beta$ satisfying this relation. 
The dealer then computes

\[
\gamma=(a_1,a_2,\dots,a_n)=\beta G .
\]

The components of $\gamma$ constitute the shares of the scheme. 
Each participant $P_i$ receives the share $a_i$ for $i\ge2$. 
The secret is implicitly encoded through the relation $s=\beta\alpha_1$.

\subsection{Reconstruction of the secret}

A subset of participants
\[
\{P_{i_1},P_{i_2},\dots,P_{i_r}\}
\]
can reconstruct the secret if and only if the column vector $\alpha_1$ is a linear combination of
\[
\alpha_{i_1},\alpha_{i_2},\dots,\alpha_{i_r}.
\]

Equivalently, let $C$ be an $[n,k;p^n]$ linear code with generator matrix $G$. 
In the corresponding SSS, suppose that the dual code $C^\perp$ contains a codeword of the form

\[
(1,0,\dots,0,b_{i_1},0,\dots,0,b_{i_r},0,\dots,0),
\]

where at least one $b_{i_j}\neq0$ and
\[
2\le i_1<i_2<\cdots<i_r\le n.
\]

Then the set of shares $a_{i_1},a_{i_2},\dots,a_{i_r}$ determines the secret. 
Indeed, such a codeword implies that

\[
\alpha_1=\sum_{j=1}^{r}x_j\alpha_{i_j},
\]

where $x_j=-b_{i_j}\in\mathbb{Z}_{p^n}$. 
Consequently, the secret can be reconstructed as

\[
s=\sum_{j=1}^{r}x_j a_{i_j}.
\]

\subsection{Example}

Consider

\[
G=
\begin{pmatrix}
3&3&1&0&0&0\\
3&1&0&1&0&0\\
3&2&0&0&1&0\\
2&3&0&0&0&1
\end{pmatrix}
\]

which generates a $[6,4]$ linear code $C$ over $\mathbb{Z}_4$. 
The dual code $C^\perp$ is generated by

\[
\begin{pmatrix}
1&0&1&1&1&2\\
0&1&1&3&2&1
\end{pmatrix},
\]

which is a minimal linear code.

The codewords of $C^\perp$ whose first coordinate equals $1$ are

\[
(1,0,1,1,1,2),\;
(1,1,2,0,3,3),\;
(1,2,3,3,1,0),\;
(1,3,0,2,3,1).
\]

Let the secret be $s=2$. 
Since
\[
\alpha_1=
\begin{pmatrix}
3\\3\\3\\2
\end{pmatrix},
\]
there are $4^3=64$ vectors $\beta$ satisfying $2=\beta\alpha_1$. 
For instance, choose

\[
\beta=
\begin{pmatrix}
2\\1\\1\\1
\end{pmatrix}.
\]

Then

\[
\gamma=\beta^T G=(2,0,2,1,1,1),
\]

and the shares distributed to the participants 
$P_2,P_3,P_4,P_5,P_6$ are

\[
\{0,2,1,1,1\}.
\]

From the minimal codewords of $C^\perp$ with first coordinate $1$, the minimal access sets are

\[
\{P_3,P_4,P_5,P_6\},\;
\{P_2,P_3,P_5,P_6\},\;
\{P_2,P_3,P_4,P_5\},\;
\{P_2,P_4,P_5,P_6\}.
\]

Each of these sets correctly reconstructs the secret using the corresponding coefficients derived from the dual codewords.

\subsection{Discussion}

If $C^\perp$ is a two-dimensional reduced minimal linear code over the field $\mathbb{F}_4$, then each minimal access set in the associated SSS contains three participants. 
However, when $C^\perp$ is defined over the ring $\mathbb{Z}_4$, each minimal access set contains four participants.

Thus, even though $\mathbb{F}_4$ and $\mathbb{Z}_4$ have the same cardinality, the combinatorial structure of access sets changes when moving from fields to rings. 
This illustrates that minimal codes over rings exhibit structural properties that differ significantly from those of finite-field codes and may lead to new classes of secret sharing schemes.

\section{Conclusion and Future Research Directions}\label{conclusion}

In this article, we investigated the construction and structural properties of two-dimensional minimal linear codes over the ring $\mathbb{Z}_{p^n}$. 
We proved that a $2\times m$ matrix generates a two-dimensional minimal linear code if and only if it contains the $p^n+p^{n-1}$ specific types of column vectors described in our construction, with
\[
m \geq p^n+p^{n-1}.
\]
If any of these essential column types is removed, the resulting code is no longer minimal. 
In particular, Lemma~\ref{lemmab}, Remark~\ref{lemma7}, and Remark~\ref{remark21} together establish a necessary lower bound on the length of a two-dimensional minimal linear code over $\mathbb{Z}_{p^n}$.

We also examined the special case $n=1$, corresponding to the field $\mathbb{Z}_p$. 
In this situation, our analysis shows that a generator matrix with exactly $p+1$ distinct column-vector types satisfies the Ashikhmin--Barg criterion for minimality. 
However, adding additional columns to such a matrix violates this condition. 
For instance, the matrix
\[
G=
\begin{pmatrix}
1 & 0 & 1 & 1 & 1 & 1 \\
0 & 1 & 1 & 2 & 2 & 2
\end{pmatrix}
\]
over $\mathbb{Z}_3$ generates a minimal linear code even though it does not satisfy the Ashikhmin--Barg condition. 
Our construction, therefore, confirms that, in the field case $n=1$, the code length must satisfy $m\ge p+1$, which coincides with the lower bound reported in~\cite{LuWuCao21}.

Furthermore, we presented applications of these constructions to secret sharing schemes derived from minimal linear codes over $\mathbb{Z}_{p^n}$. 
The obtained access structures illustrate how minimal codes over rings can exhibit different combinatorial properties than their counterparts over finite fields.

Finally, all theoretical results established in this work were verified through explicit computations performed using the \textsc{Magma} algebra system. 
These experiments confirm the structural constraints governing minimal codes over $\mathbb{Z}_{p^n}$ and highlight several differences from the classical field setting.

Future research directions include the study of higher-dimensional minimal linear codes over finite rings, the characterization of their access structures in secret sharing schemes, and the investigation of potential applications in ring-based cryptographic constructions.

\bibliographystyle{plain}

\end{document}